\newcounter{myctr}
\def\myitem{\refstepcounter{myctr}\bibfont\noindent\ifnum\themyctr>9\else\phantom{0}\fi\hangindent17pt\themyctr.\enskip}

\documentclass{ws-ijqi}

\usepackage{hyperref}
\usepackage[super,sort,compress]{cite}
\usepackage{url}
\usepackage[usenames]{color}

\usepackage{xcolor}

\newcommand{\ket}[1]{\ensuremath{\left|{#1}\right\rangle}}
\newcommand{\bra}[1]{\ensuremath{\left\langle{#1}\right |}}

\newcommand{\Trr}[1]{\textrm{Tr}\!\left[#1\right]}

\newcommand{\ba}[1]{
	\begin{align}
		#1
\end{align}}

\newcommand{\av}[1]{\overline{#1}}

\begin{document}

\catchline{}{}{}{}{}

\title{R\'enyi-Holevo inequality from $\alpha$-$z$-R{\'e}nyi relative entropies}

\author{Diego G. Bussandri$^{1,2}$, Grzegorz Rajchel-Mieldzio{\'c}$^{3}$, Pedro W. Lamberti$^{2,4,5}$,\\
 Karol {\.Z}yczkowski$^{5,6}$}

\address{$^1$Instituto  de  F\'isica  La  Plata  (IFLP)  and  Departamento  de  F\'isica,  Facultad  de Ciencias Exactas, Universidad Nacional de La Plata, C.C. 67, 1900 La Plata, Argentina}

\address{$^2$Consejo Nacional de Investigaciones Científicas y Técnicas de la República Argentina (CONICET), Av. Rivadavia 1917, C1033AAJ, CABA, Argentina}

\address{$^3$Institut de Ciencies Fotoniques (ICFO), The Barcelona Institute of Science and Technology, 08860 Castelldefels, Barcelona, Spain}

\address{$^4$Facultad de Matem\'atica, Astronom\'{\i}a, F\'{\i}sica y Computaci\'on, Universidad Nacional de C\'ordoba, \\ Av. Medina Allende s/n, Ciudad Universitaria, X5000HUA C\'ordoba, Argentina}
\address{$^5$Faculty of Physics, Astronomy and Applied Computer Science, Institute of Theoretical Physics, Jagiellonian University, ul. {\L}ojasiewicza 11, 30-348, Krak{\'o}w, Poland}
\address{$^6$Center for Theoretical Physics, Polish Academy of Sciences, Al. Lotnik{\'o}w 32/46, 02-668 Warszawa, Poland}

\maketitle


\begin{abstract}

We investigate
bounds in the transmission of classical information through quantum systems. Our focus lies in the generalized Holevo theorem, which provides a single-letter Holevo-like inequality from arbitrary quantum distance measures. Through the introduction of the $\alpha$-$z$-R{\'e}nyi relative entropies, which comprise known relevant quantities such as the R{\'e}nyi relative entropy and the sandwiched R{\'e}nyi relative entropy, we establish the Holevo-R{\'e}nyi inequality. This result leads to a quantum bound for the $\alpha$-mutual information, suggesting new insights into communication channel performance and the fundamental limits for reliability functions in memoryless multi-letter communication channels.

\end{abstract}


\markboth{D. G. Bussandri, G. Rajchel-Mieldzio{\'c}, P. W. Lamberti, K. {\.Z}yczkowski}
{R{\'e}nyi-Holevo inequality from $\alpha$-$z$-R{\'e}nyi relative entropies}


 \vskip 0.5 cm
 {\sl Dedicated to Alexander S. Holevo on the occasion of his 80-th birthday.}

\section{Introduction}\label{sec:intro}

Quantum information theory explores the fundamental limits and possibilities of transmitting and processing information using quantum systems~\cite{Wilde2017a}. The Holevo Theorem\cite{Holevo1973} is a fundamental result in this field that provides insights into the transmission of classical information through a quantum communication setup. In Ref. \refcite{bussandri2020generalized}, the \textit{generalized Holevo theorem} was introduced, which rewrites the one-letter Holevo inequality in terms of quantum distance measures. Current contribution focuses on exploring the implications of the generalized Holevo theorem and extends its findings using quantum R{\'e}nyi divergences\cite{Audenaert2015}.

In this context, the sender -- Alice -- aims to transmit the value of a discrete random variable $X$ to the receiver -- Bob -- using quantum states. The communication is achieved by encoding the classical outcomes of $X$ into a specific quantum state, which is then transmitted to Bob. 
Finally, Bob retrieves the encoded classical information by measuring suitable observables of the quantum system. The generalized Holevo theorem\cite{bussandri2019generalized} establishes a quantum bound for an arbitrary
 (yet suitable, see Theorem \ref{theo:Bussandri2020}) 
 classical distance measures between the joint probability distribution of $X$ (the classical message) and $Y$, standing for the random variable defined by Bob's measurements, 
 and the corresponding separable distribution.

In this work, we introduce the concept of R{\'e}nyi divergences and its extensions to the quantum state space: the R{\'e}nyi relative entropy\cite{Petz1986}, the sandwiched R{\'e}nyi relative entropy\cite{Wilde2014} and 
the $\alpha$-$z$-R{\'e}nyi relative entropies\cite{Audenaert2015}. By exploiting its properties, we obtain the main result of the paper which we shall refer to as the \textit{Holevo-R{\'e}nyi inequality}, leading to a bound for the $\alpha$-\textit{mutual information}\cite{sibson1969information}. This result provides new insights into the performance of communication channels and helps to deepen the understanding of the fundamental limits of classical information transmission through quantum means\cite{burnashev1998reliability} by establishing new bounds for the Gallager's reliability function for memoryless multi-letter communication channels --
see for instance Ref. \refcite{Takeoka2010}.
 This function allows for a finer measure of the quality of the channel because it provides a lower bound for the largest possible rate of decay of the error probability~\cite{Alsan2015}.


The paper is organized as follows: In Sec. \ref{sec:motivation} we recall the
 generalized Holevo theorem. Sec. \ref{sec:RényiDist} contains the main results and definitions regarding quantum extensions of the $\alpha$-R{\'e}nyi divergence. In Sec. \ref{sec:RHB} we establish the Holevo-R{\'e}nyi inequality providing its tightest bounds. Finally,  in Sec. \ref{sec:App} we obtain bounds for the reliability functions and error exponents in the multi-letter memoryless case.

\section{Generalized Holevo Theorem}
\label{sec:motivation}

Consider a discrete variable $X$ taking on values from alphabet $\mathcal{X}$, randomly distributed according to the probability distribution $p(x)$, where $x$ is a letter in alphabet $\mathcal{X}$. In a quantum communication setup, the values of $X$ are transmitted using a set of quantum states (also known as \textit{signals}), $\{\rho_x\}_{x\in\mathcal{X}}$, belonging to $B(\mathcal{H})$: the semi-definite positive operators with unit trace defined on a finite-dimensional Hilbert space $\mathcal{H}$. These quantum signals are prepared utilizing macroscopic devices~\cite{Holevo2019}: By adjusting their parameters one can modify the quantum state, providing the possibility to encode a classical outcome $X=x$ into a specific quantum state $\rho_x$, which is sent to the receiver Bob. The correspondence $x \to \rho_x$ defines a classical-quantum channel.

To retrieve the encoded classical information from the output states, Bob performs a POVM (positive operator-valued measure) $\mathcal{M}$ given by positive operators $\{M_y\}_{y\in\mathcal{Y}}$. We shall refer to the random variable defined by the measurement results as $Y$, taking values from the alphabet $\mathcal{Y}$. The conditional probability of obtaining the result $Y=y$, given the input signal $\rho_x$, is 
\begin{align}
    W^M_{sq}(y|x)\doteq \Trr{\rho_xM_y}. \label{eq:conditionalprob-channel}
\end{align}
The joint probability distribution, and the corresponding \textit{separable} distribution, between $X$ and $Y$ are, respectively,
\begin{align}
    p_M(x,y)&\doteq p(x)W^M_{sq}(y|x),\label{eq:jointprob}\\
    (p\times q)(x,y)&\doteq p(x)q_M(y), \label{eq:separableprobdistr}
\end{align}
with $q_M(y)=\sum_{x\in\mathcal{X}} p_M(x,y)$. 

The following theorem, established in Ref.~\refcite{bussandri2020generalized}, provides a bound for the \textit{distance} between the joint probability $P$ and separable $p\times q$:
\begin{theorem}\label{theo:Bussandri2020}
	Let $d(\cdot, \cdot)$ be a monotone distance measure~\footnote{A functional $d(\cdot,\cdot)$ defined over $B(\mathcal{H})$ -- the set of density operators on the finite-dimensional Hilbert space $\mathcal{H}$, is a monotone distance measure if $d(\rho,\sigma)\geq 0$ for all $\rho$ and $\sigma$ belonging to $B(\mathcal{H})$ and $d(\rho,\sigma)\geq d[\Phi(\rho),\Phi(\sigma)]$, with $\Phi$ an arbitrary completely positive trace-preserving map on $B(\mathcal{H})$.} between quantum states which additionally satisfies the flag condition \footnote{The flag condition is:	$d(\rho_1,\rho_2)=\sum_k p_k d(\sigma_k,\tau_k)$, with $\rho_1$ and $\rho_2$ being two block-diagonal (or classical-quantum) quantum states such that $\rho_1=\sum_k p_k \ket{k}\bra{k}\otimes\sigma_k$ and $\rho_2=\sum_k p_k \ket{k}\bra{k}\otimes\tau_k$.  The set $\{\ket{k}\bra{k}\}_k$ is an orthonormal base of the corresponding Hilbert space.},
	then holds the following
	\ba{D(P,p\times q)\leq \sum_{x\in\mathcal{X}} p(x) d(\rho_x,\av{\rho}),\label{eq:GenHolIneq}}
where $D(\cdot||\cdot)$ denotes the classical analogue of $d(\cdot,\cdot)$, in the space of probability distributions, and $\av{\rho}=\sum_{x\in\mathcal{X}}p(x)\rho_x$.
\end{theorem}
The distance measures $d(\cdot,\cdot)$ provide alternative ways of quantifying different aspects of the distinguishability of quantum states in the same way as the previous theorem establishes bounds for a variety of quantifiers that characterize the transmission of classical information by using the ensemble of quantum states $\{p(x),\rho_x\}_{x\in\mathcal{X}}$. 
The most famous case is given by the quantum relative entropy $S_r(\rho,\sigma)=\Trr{\rho (\log_2\rho-\log_2\sigma)}$ of Umegaki~\cite{Umegaki_1962}, which reduces to the Shannon relative (or Kullback-Leibler-Sanov~\cite{Holevo2019}) entropy $H_r$, if $\rho$ and $\sigma$ commute.
The generalized Holevo bound reduces to the standard form for $d=S_r$, 
which is a fundamental result in the field of quantum information theory establishing an upper bound for the quantity of classical information that can be transmitted by encoding it into an ensemble of quantum states. 

For completeness, let us recall the, now standard, Holevo's bound~\cite{Holevo1973}.
\begin{theorem}[Holevo's bound\cite{Holevo1973}] The accessible information $I_a(p)$ can be bounded as:
\ba{\label{eq:Holevoinfo}
I_a(p)\doteq\max_{\mathcal{M}} I(X,Y)\leq \sum_{x\in\mathcal{X}} p(x) S_r(\rho_x||\av{\rho})\doteq C(p),} \label{theo:Holevo}
where
\ba{\label{eq:mutualinfo}
I(X,Y)\doteq H(X)+H(Y)-H(X,Y),
}
with $H(X)\doteq -\sum_{x\in \mathcal{X}} p(x) \log_2 p(x)$ and $H(Y)\doteq -\sum_{y\in\mathcal{Y}} q_M(y) \log_2 q_M(y)$ the Shannon's entropies of $X$ and $Y$, respectively, and $H(X,Y)\doteq -\sum_{x\in \mathcal{X}} p_M(x,y) \log_2 p_M(x,y)$ the joint entropy.
\end{theorem}
The \textit{accessible information} $I_a(p)$ is the maximal amount of information which can be sent by using a quantum encoding, namely, it stands for the maximum information that Bob can obtain by performing measurements in the one-shot case. The upper bound $C(p)$ is known as \textit{Holevo information }and it can be rewritten in the following way,
\begin{align}
C(p) = \sum_{x\in\mathcal{X}} p(x) S_r(\rho_x||\av{\rho}) = S(\av{\rho})-\sum_{x\in\mathcal{X}}p(x)S(\rho_x),
\end{align}
where $S(\rho)=-\Trr{\rho \log_2 \rho}$ is the von Neumann entropy.

\section{Classical and quantum R\'enyi divergences}\label{sec:RényiDist}
	
\noindent The $\alpha$-R\'enyi divergence between two probability distributions $p$ and $q$ 
on a finite alphabet $ \mathcal{X}$ is defined as\cite{VanErven2014}
\ba{\label{eq:ClassicalRenyiDiv}
    D_\alpha(p||q)\doteq\frac{1}{\alpha-1}\log\sum_{x\in \mathcal{X}} p(x)^\alpha q(x)^{1-\alpha},
}
with R{\'e}nyi parameter $\alpha>0$. There is a natural one-parameter extension of $D_\alpha$ to the quantum state space, which we refer to as the R{\'e}nyi relative entropy~\cite{Petz1986} (RRE):
	\ba{
		d_\alpha(\rho||\sigma)\doteq\frac{1}{\alpha-1}\log \Trr{\rho^\alpha \sigma^{1-\alpha}},
	\label{eq:RenyiRelativeEntropy}}
for $\rho,\sigma \in  B(\mathcal{H})$ such that $\textrm{supp}  \ \rho \subseteq \textrm{supp} \ \sigma$ (i.e. $\sigma \! \gg \! \rho$), or $\rho \not\perp \sigma$. Here, parameter $\alpha\in(0,1)$. 
We define a set of pairs of quantum states that meet the above conditions by symbol $\mathbb{B}$.

The sandwiched R{\'e}nyi relative entropy represents a different non-commutative extension of the classical $\alpha$-R\'enyi divergence~\cite{Wilde2014}:
	\ba{\label{eq:RenyiSandwiched}
		\tilde{d_\alpha}(\rho||\sigma)\doteq \frac{1}{\alpha-1}\log \Trr{\left(\sigma^{\frac{1-\alpha}{2\alpha}}\rho \sigma^\frac{1-\alpha}{2\alpha}\right)^\alpha},
	}  
for $(\rho,\sigma)\in\mathbb{B}$. 
The functionals defined by Eqs. \eqref{eq:RenyiRelativeEntropy} and \eqref{eq:RenyiSandwiched} are both useful in different areas of quantum information theory \cite{mosonyi2009generalized,Lin2015}. The term \textit{non-commutative extensions} refers to the fact that if $\rho$ and $\sigma$ commute, both R{\'e}nyi relative entropy and the sandwiched R{\'e}nyi relative entropy reduce to Eq. \eqref{eq:ClassicalRenyiDiv}. 
Furthermore, they connect to another important notion, namely the Uhlmann fidelity~\cite{Uhlmann1976,Jozsa1994},
$F(\rho,\sigma) = (\Trr{\sqrt{\sigma^{1/2} \rho \sigma^{1/2}}})^2$. 
More specifically, for $\alpha=1/2$, Eq. \eqref{eq:RenyiSandwiched} takes the form of $d_{1/2}(\rho||\sigma)= - \log F$.
The right-hand side of the latter equation, called log-fidelity, forms a semimetric, 
as it does not obey the triangle inequality. Observe that
by varying the parameter $\alpha$ from $1/2$ to $1$, the quantity $d_{\alpha}$ 
provides a continuous transition between the log fidelity (symmetric in both arguments) and the quantum relative entropy of Umegaki~\cite{Umegaki1962}, so the permutation symmetry breaks down with the increase of $\alpha$. 
	
In Ref. \refcite{Audenaert2015}, a new family of quantum R\'enyi divergences, the $\alpha$-$z$-R\'enyi relative entropies, was introduced as follows:
 	\ba{
		d_{\alpha,z}(\rho||\sigma)&\doteq \frac{1}{\alpha-1} \log f_{\alpha,z}(\rho||\sigma), \hspace{.5cm} \textrm{where} \nonumber \\\label{eq:alphaZrenyi}
		f_{\alpha,z}(\rho||\sigma)&\doteq\Trr{(\rho^{\alpha/2z}\sigma^{(1-\alpha)/z}\rho^{\alpha/2z})^z}, \text{ with } \rho,\sigma\in\mathbb{B}.
	}
Note the connection linking it to both R\'enyi relative entropy and the sandwiched quantum R\'enyi divergence: $d_{\alpha,1}=d_\alpha$ and $d_{\alpha,\alpha}=\tilde{d_\alpha}$. We shall refer to $f_{\alpha,z}(\rho||\sigma)$ as the $\alpha$-$z$-R\'enyi overlap, in an analogy to a similar notation used from Ref.~\refcite{Fuchs1996a}.

The following lemma states the conditions under which the data processing inequality holds for the $\alpha$-$z$-R\'enyi relative entropies:	
	\begin{lemma}[Data processing inequality for $\alpha$-$z$-RRE]\label{theo:Monotonicity}
Let $\rho,\sigma\in \mathbb{B}$ be two arbitrary \textit{invertible} density matrices (i.e. positive definite), and $\Phi: B(\mathcal{H})\to  B(\mathcal{H})$ be an arbitrary completely positive trace-preserving quantum map. It follows that
 \ba{d_{\alpha,z}[\Phi(\rho)||\Phi(\sigma)]\leq d_{\alpha,z}(\rho||\sigma),\label{eq:dpi}}
if and only if one of the following conditions is met~\cite{Zhang2020a,Zhang2020}
\begin{enumerate}
\item $\alpha<1$ and $z\geq \max \{\alpha, 1-\alpha\}$,\label{eq:ineq1}
\item $1<\alpha\leq 2$, and $\alpha/2\leq z \leq \alpha$, 
\item $2\leq\alpha<\infty$ and $\alpha-1\leq z \leq \alpha$.\label{eq:ineq3}
 \end{enumerate}
\end{lemma}
We shall refer to the set of all pairs $(\alpha,z)\in\mathbb{R}^2$ fulfilling the inequalities \eqref{eq:ineq1}-\eqref{eq:ineq3} as $\mathbb{M}$, see Fig. \ref{fig:DPIvalues} for a graphical visualization. 

\begin{figure}
    \centering
    \includegraphics[width=0.7\textwidth]{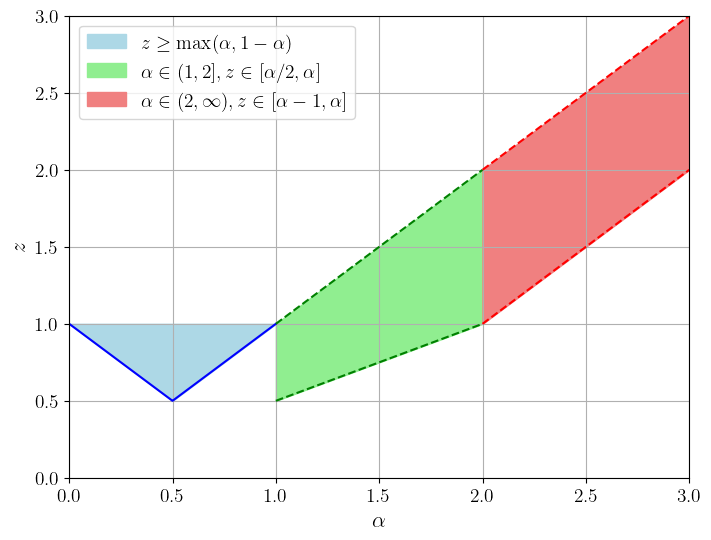}
    \caption{Region plot of the parameter space
     $\mathbb{M}$: Set of all pairs $(\alpha,z)\in\mathbb{R}^2$ satisfying the data processing inequalities, Eq. \eqref{eq:dpi}.}
    \label{fig:DPIvalues}
\end{figure}

The order of the $\alpha$-$z$-R{\'e}nyi relative entropies with respect to parameter $z$ is ruled by the following Lemma, demonstrated in Ref. \refcite{Lin2015}: 

\begin{lemma}\label{theo:Ordering}
For arbitrary quantum states $(\rho,\sigma)\in \mathbb{B}$, the function defined by $$z \longmapsto d_{\alpha,z}(\rho||\sigma),$$ is monotonically decreasing for $\alpha>1$ and monotonically increasing for $\alpha<1$.
\end{lemma}

\section{Results: R\'enyi--Holevo inequalities} \label{sec:RHB}
	
Let us return to the quantum communication framework, explained in Sec. \ref{sec:motivation}, and focus on the $\alpha$-R\'enyi divergence between the joint and separable probability distributions given by Eqs. \eqref{eq:jointprob} and \eqref{eq:separableprobdistr}, respectively:
	\ba{
 D_\alpha(P||p\times q)=\frac{1}{\alpha-1}\log \sum_{x\in\mathcal{X},y\in\mathcal{Y}} p_M(x,y)^{\alpha} [p(x)q_M(y)]^{1-\alpha}. 
	}
If we maximize this quantity over the set of possible measurements, we will obtain a \textit{R\'enyi accessible information}, which coincides with the usual accessible information if we take the limit $\alpha \to 1$. 
	
Now, let us explore the quantum bounds for $D_\alpha(P||p\times q)$ resulting from the introduction of the $\alpha$-$z$-R{\'e}nyi relative entropy, Eq. \eqref{eq:alphaZrenyi}, in Theorem \ref{theo:Bussandri2020}. 
Therefore, we have to study whether monotonicity and the flag condition are satisfied for the $\alpha$-$z$-R{\'e}nyi relative entropies. 
Since, in general, the flag condition does not hold for $d_{\alpha,z}(\cdot||\cdot)$, we define instead an alternate distance measure, that will satisfy the flag condition for $(\alpha,z)\in \mathbb{M}$,
\begin{align}\label{eq:alternateDist}
d^A_{\alpha,z}(\rho,\sigma)\doteq\begin{cases}
f_{\alpha,z}(\rho||\sigma) -1 \text{, for } \alpha > 1, \\
1- f_{\alpha,z}(\rho||\sigma) \text{, for } \alpha < 1,
\end{cases}
\end{align}
so that the conditions of Theorem \ref{theo:Bussandri2020} are met. The following theorem provides bounds for the R{\'e}nyi accessible information in the single-letter case:
\begin{theorem}[Holevo-R\'enyi inequality]\label{theo:MAIN}
 Let $(\alpha,z)\in \mathbb{M}$ and $f_{\alpha,z}(\rho||\sigma)=\text{\emph{Tr}}\!\left[\sigma^{\frac{(1-\alpha)}{2z}}\rho^{\frac{\alpha}{z}}\sigma^{\frac{(1-\alpha)}{2z}}\right]^z$. Then the following holds,
\begin{align}\label{eq:MainTheo}
\sum_{x\in\mathcal{X},y\in\mathcal{Y}} p_M(x,y)^\alpha [p(x)q_M(y)]^{1-\alpha} \begin{cases}
\leq \sum_x p(x) f_{\alpha,z}(\rho_x||\av{\rho}) \text{, for } \alpha > 1, \\
\geq \sum_x p(x) f_{\alpha,z}(\rho_x||\av{\rho}) \text{, for } \alpha < 1.
\end{cases}
\end{align}
with $\av{\rho}=\sum_{x\in\mathcal{X}}p(x)\rho_x$.

	\end{theorem}

 \begin{proof}
The proof is an immediate consequence of inserting $d^A_{\alpha,z}$ into Theorem \ref{theo:Bussandri2020} as a quantum distance measure, and since the monotonicity and the flag condition hold. 

The behavior of $\alpha$-$z$-R{\'e}nyi overlap $f_{\alpha,z}(\cdot||\cdot)$ under arbitrary completely positive trace-preserving map $\Phi$ is characterized by Lemma \ref{theo:Monotonicity}: For $(\alpha,z)\in \mathbb{M}$ and $\alpha<1$, it holds that $f_{\alpha,z}[\Phi(\rho)||\Phi(\sigma)]\geq f_{\alpha,z}(\rho||\sigma)$; while if $\alpha>1$, we have the opposite $f_{\alpha,z}[\Phi(\rho)||\Phi(\sigma)]\leq f_{\alpha,z}(\rho||\sigma)$. 
The relation $f_{\alpha,z}(\rho||\rho)=1$ for all $(\alpha,z)\in\mathbb{M}$,
implies that $d^A_{\alpha,z}(\cdot,\cdot)$, defined in Eq. \eqref{eq:alternateDist}, is a proper monotone quantum distance since it satisfies the data processing inequality, 
and it is positive for all $\rho,\sigma \in \mathbb{B}$: 
$$d^A_{\alpha,z}(\rho,\sigma)\geq d^A_{\alpha,z}[\Phi_d(\rho),\Phi_d(\sigma)] = d^A_{\alpha,z}(\frac{\mathbb{I}}{\dim \mathcal{H}},\frac{\mathbb{I}}{\dim \mathcal{H}})=0,$$
with $\Phi_d(\cdot)$ being the completely depolarizing map that transforms any state $\rho$ to the maximally mixed one, $\Phi_d(\rho)=\mathbb{I}/\dim \mathcal{H}$. 
Here, the inequality holds because $d^A_{\alpha,z}$ fulfills the data processing inequality.  

The flag condition, necessary to be satisfied by $d^A_{\alpha,z}$ for complying with the assumptions of Theorem \ref{theo:Bussandri2020}, can be easily verified by checking:
$$\rho_{cq}^{\lambda}=\sum_k p_k^\lambda \ket{k}\bra{k}\otimes \rho_k^\lambda,$$
for $\lambda\in\mathbb{R}$ and $\rho_{cq}=\sum_k p_k \ket{k}\bra{k}\otimes \rho_k$ being an arbitrary classical-quantum state~\cite{Wilde2017a}.
Finally, using Theorem \ref{theo:Bussandri2020} for $d^A_{\alpha,z}$, we arrive at Eq. \eqref{eq:MainTheo}.
 \end{proof}

By taking logarithm of both sides of inequality \eqref{eq:MainTheo} and then applying the corresponding concavity, we
see that the $\alpha$-$z$-R{\'e}nyi relative entropy leads to a similar bound as the relative entropy $S_r(\cdot||\cdot)$ does in Holevo's theorem -- see Theorem \ref{theo:Holevo},
\begin{corollary} For $(\alpha,z)\in\mathbb{M}$ and $\alpha < 1$, the $\alpha$-$z$-R{\'e}nyi divergence can be bounded as
	\ba{ 
	D_\alpha(P||p\times q) \leq \sum_i p_i d_{\alpha,z}(\rho_i||\av{\rho}),
}
where the notation is analogous to that of Theorem~\ref{theo:Bussandri2020}.
\end{corollary}

As we can see, the left-hand side in the Holevo-R{\'e}nyi inequality does not depend on $z$, therefore, by simply employing Lemma~\ref{theo:Ordering} we can verify the tightest bounds for $(\alpha,z)\in\mathbb{M}$.  For the rest of this proof, we will search for the optimal value of $z$ leading to these optimal bounds. 

From Lemma~\ref{theo:Ordering}, we obtain that $z \mapsto f_{\alpha,z}$ is monotonically decreasing with $z$ for all $\alpha>0$. If $\alpha <1$, we have to maximize $f_{\alpha,z}$ as a function of $z$, taking correspondingly the minimum $z$ such that $(\alpha,z)\in\mathbb{M}$ (see Fig. \ref{fig:DPIvalues}), that is,
$$z^*_{\alpha<1}=\begin{cases}
   \alpha, \text{ if } \alpha\in[1/2,1),\\
   1-\alpha, \text{ if } \alpha\in(0,1/2).
\end{cases}$$

If $\alpha>1$, we have to take the lowest value of $f_{\alpha,z}$, this corresponds to the largest value of $z$ such that $(\alpha,z)\in\mathbb{M}$, therefore, $z^*_{\alpha>1}=\alpha$, see Figure \ref{fig:DPIvalues}. In summary, the tightest bounds in Eq. \eqref{eq:MainTheo} are given by:
\begin{align}\label{eq:tightestbounds}
f_{sq}(p,\alpha)\doteq \begin{cases}
     \sum_x p(x) \Trr{\left(\rho_x\overline{\rho}^{\frac{1-\alpha}{\alpha}}\right)^\alpha} \text{ for } \alpha \geq 1/2, \\
\sum_x p(x)
\Trr{(\rho_x^{\frac{\alpha}{1-\alpha}}\overline{\rho})^{1-\alpha}} \text{ for } \alpha < 1/2.
\end{cases}
\end{align}
with
\begin{align}\label{eq:tightestbounds2}
    D_\alpha(P||p\times q) \leq \frac{1}{\alpha-1} \log f_{sq}(p,\alpha),
\end{align}
for all $\alpha\in(0,1)\cup(1,\infty)$.

Finally, the $\alpha$-$z$-R{\'e}nyi overlap corresponding to the sandwiched R\'enyi divergence $f_{\alpha,\alpha}$ leads to the tightest bound for all $(\alpha,z)\in\mathbb{M}$ except for $\alpha<1/2$, for which the optimal $z$ is given by $z=1-\alpha$. This case corresponds to the distance measure $d_{\alpha,1-\alpha}(\rho||\sigma)$, referred to as \textit{reverse sandwiched R\'enyi relative entropy}, see Ref. \refcite{Audenaert2015},
	\ba{
		d_{\alpha,1-\alpha}(\rho||\sigma)=\frac{\alpha}{(1-\alpha)}\tilde{d}_{1-\alpha}(\sigma||\rho).
	}

 \section{Applications}\label{sec:App}

In classical information theory, there are several complementary proposals to the mutual information $I(X,Y)$, which represents the standard quantifier of correlation degree between two random variables $X$ and $Y$. Among these, the R{\'e}nyi generalisations of the mutual information introduced by Arimoto~\cite{arimoto1977information}, Csisz{\'a}r~\cite{Csiszar1995}, and Sibson~\cite{sibson1969information}, are particularly noteworthy~\cite{Csiszar1995}, with the Sibson's proposal, referred to as $\alpha$-\textit{mutual information} $I_\alpha$, being the most promising one~\cite{Verd}. In the case of an arbitrary channel $W(y|x)$, $I_\alpha$ is defined as: 
\begin{align}
    I_\alpha(p,W)\doteq \frac{\alpha}{\alpha-1}\log \sum_{y\in\mathcal{Y}} \left[\sum_{x\in \mathcal{X}} p(x) W(y|x)^{\alpha}\right]^{1/\alpha}. 
\end{align}
The $\alpha$-mutual information was originally conceived as a measure of diversity between conditional probability distributions $\{W(\cdot|x)\}_{x\in\mathcal{X}}$ under the name \textit{information radius of order} $\alpha$, see Ref. \refcite{Csiszar1995}. 
In general, $I_\alpha(p,W)$ quantifies different features on the relation between $X$ and $Y$, and the channel $W(y|x)$~\cite{Verd}. A particular relevant quantity is the channel capacity of order $\alpha$:
\begin{align}\label{eq:capacityalpha}
    C_\alpha(W) = \max_{p}I_\alpha(p,W),
\end{align}
which has been recently used in the context of resource theory of measurement informativeness~\cite{Skrzypczyk2019,Ducuara2022} with the channel $W(y|x)$ defined as in Sec. \eqref{sec:motivation}.

The Holevo-R{\'e}nyi inequality, see Theorem \ref{theo:MAIN}, restricts the behaviour of the $\alpha$-mutual information when classical information is encoded using a set of quantum states. In this context, the communication channel is defined as $W(y|x)=W^M_{sq}(y|x)=\Trr{M_y\rho_x}$, as shown in Eq. \eqref{eq:conditionalprob-channel}. Utilizing Eq. \eqref{eq:tightestbounds2} and considering that the $\alpha$-mutual information satisfies the relation~\cite{Csiszar1995}:
\begin{align}
I_\alpha(p,W) = \min_q D_\alpha (P||p\times q),
\end{align}
we can establish the following bound:
\begin{align}\label{eq:boundOnAlphaMutualInfo}
I_\alpha(p,W^M_{sq}) \leq \frac{1}{\alpha-1} \log f_{sq}(p,\alpha),
\end{align}
where $f_{sq}(p,\alpha)$ is given by Eq. \eqref{eq:tightestbounds}. Consequently, the Holevo-R{\'e}nyi inequality provides quantum bounds for the $\alpha$-mutual information and, thereby, for the channel's capacity of order $\alpha$, Eq. \eqref{eq:capacityalpha}, in the one-shot case.

The aforementioned inequality yields a specific bound in the context of reliability functions and error exponents for multi-letter communication channels (i.e., composite channels). Consider the scenario where we aim to send $n$-lettered \textit{codewords}, denoted as $w=(x_1,\cdots,x_n)$, through $n$ uses of the classical-quantum channel $x \to \rho_x$, discussed in Sec. \ref{sec:motivation}. We establish therefore the correspondence:
$$w \to \rho_w=\rho_{x_1}\otimes \cdots \otimes \rho_{x_n}\in B(\mathcal{H}^{\otimes n}).$$
A \textit{codebook} of size $N$ consists of a set of codewords $\mathcal{W}=\{w^i\}_{i=1}^N$. The receiver performs measurements on the quantum states $\rho_w$ to extract information about the transmitted codeword. This defines a \textit{code}, which is a collection of $N$ pairs $(w^i,D_i)$, where $\mathcal{D}=\{D_i\}_{i=1}^N$ is a joint POVM, and the entire procedure constitutes a composite channel to transmit codewords.

The probability of correctly guessing the codeword $w^i$ is given by $\Trr{D_i \rho_{w^i}}$. Consequently, the mean probability of error, indicating the probability of guessing the wrong codeword, is:
\begin{align}\label{eq:proberrQ}
    \overline{p}^q_e(\mathcal{W},\mathcal{D})=\frac{1}{N}\sum_{i=1}^N (1-\Trr{D_i \rho_{w^i}}).
\end{align}
The capacity $\tilde{C}$ of the classical-quantum channel $x\to \rho_x$ is such that the minimal probability of error,
\begin{align}\label{eq:proberrQmin}
\overline{p}^q_e(2^{nR},n)=\min_{\mathcal{W},\mathcal{D}}\overline{p}^q_e(\mathcal{W},\mathcal{D}),
\end{align}
tends to zero for any transmission rate $0\leq R<\tilde{C}\leq 1$ and does not go to zero for $1\geq R>\tilde{C} \geq 0$.

Remarkably, the single-letter Holevo Theorem (Thm.~\ref{theo:Holevo}) provides an inequality between two different kinds of capacity $\tilde{C}$: On one hand, for arbitrary POVM measurements $\mathcal{D}$, the capacity of the classical-quantum channel $x\to \rho_x$ is given by the maximum of the Holevo information $C(p)$ over the input probability distributions $p$, i.e., $\tilde{C}_q=\max_p C(p)$, as defined in Eq. \eqref{eq:Holevoinfo}. On the other hand, for measurements $\mathcal{D}$ restricted to product operators $D_i=M_{x^i_1}\otimes \cdots \otimes M_{x^i_n}$ (i.e., performing individual measurements on the output states $\rho_{x_j}$ instead of on the entire state $\rho_w$), the channel capacity is given by the maximum of accessible information $I_a(p)$, i.e., $\tilde{C}_{sq}=\max_p I_a(p)$. This scenario, in which product measurement operators are employed, is generally referred to as \textit{semi-quantum}~\cite{Ban1998}. We shall denote as $\overline{p}_e^{sq}(2^{nR},n)$ the probability of error defined in Eq. \eqref{eq:proberrQmin} corresponding to the semi-quantum case, i.e. where the minimization is taken over the product POVMs.

The \textit{reliability function} $E(R)$ stands for a more precise assessment of the channel performance~\cite{Alsan2015} than the usual channel capacity $\tilde{C}$, on the grounds that $E(R)$ characterizes the exponential decay of the error probability as $n\to\infty$, for $R<\tilde{C}$. Specifically, when we consider the semi-quantum case (product measures), the composite channel is \textit{memoryless} and therefore the reliability function is~\cite{Csiszar1995}:
\begin{align}
    E_{sq}(R)&= \min_{0\leq s \leq 1} \{\max_p[E^{sq}_0(s)]-sR\}, \text{ with }\\ \label{eq:gallagerfunc}
    E^{sq}_0(s)&=-\log \sum_{y\in\mathcal{Y}} \left[\sum_{x\in \mathcal{X}} p(x) W_{sq}^M(y|x)^{\frac{1}{1+s}}\right]^{1+s},
\end{align}
satisfying
\begin{align}
    \overline{p}^{sq}_e(2^{nR},n)\leq \exp[-nE_{sq}(R)].
\end{align}
In this context, the Holevo-R{\'e}nyi
 inequality establishes a bound for the reliability function in the memoryless case:

\begin{proposition}[Reliability function bound]

Let $s\in(0,1]$ and let $E^{sq}_0(s)$ be the Gallager function defined in Eq. \eqref{eq:gallagerfunc}.
Then the following inequality holds:
\begin{align}\label{eq:SemiQuantumGallagerFunc}
E^{sq}_0(s)\leq -(1+s)\log \sum_{x\in\mathcal{X}}p(x)\text{\emph{Tr}}\! \left[ {\left(\rho_x\overline{\rho}^s\right)^{\frac{1}{1+s}}} \right] \doteq \mathcal{E}(s).
\end{align}
\end{proposition}
To demonstrate the previous proposition we have to take $\alpha=1/(1+s)$ in Eq. \eqref{eq:boundOnAlphaMutualInfo}. As in this case $\alpha\in [1/2,1)$, we find that the corresponding bound is $f_{sq}(p,\alpha)=\sum_x p(x) \Trr{\left(\rho_x\overline{\rho}^{\frac{1-\alpha}{\alpha}}\right)^\alpha}$, because of Eq. \eqref{eq:tightestbounds}.

\vspace{.4cm}

The reliability function in the semi-quantum case was already addressed in the literature by Helstrom~\cite{helstrom1987cut}, Charbit~\cite{charbit1989cutoff}, Bendjaballah~\cite{Bendjaballah1989} and more recently by Takeoka et al.~\cite{Takeoka2010}. 
In Ref. \refcite{charbit1989cutoff}, a similar upper bound to Eq. \eqref{eq:SemiQuantumGallagerFunc} can be found for the cutoff rate $E^{sq}_0(1)$ when the states $\{\rho_x\}_x$ are pure:
\begin{align}\label{eq:cutoffSQ}
    E^{sq}_0(1)\leq -\log \min_p \sum_{x,\tilde{x}\in\mathcal{X}} p_x p_{\tilde{x}} \sqrt{\Trr{\rho_x\rho_{\tilde{x}}}} \doteq \tilde{E}^{sq}_0(1).
\end{align}

In the case of joint measurements and in the particular case of pure signal states $\{\rho_x\}_x$, it was demonstrated that~\cite{burnashev1998reliability}:
\begin{align}
    \overline{p}^q_e(2^{nR},n)\leq \exp[-nE_{q}(R)],
\end{align}
where
\begin{align}
    E_{q}(R)&= \min_{0\leq s \leq 1} \{\max_p[E^q_0(s)]-sR\}, \text{ with }\\ \label{eq:gallagerfuncBH}
    E^q_0(s)&=-\log \Trr{\left(\sum_{x\in\mathcal{X}}p(x)\rho_x^{\frac{1}{1+s}}\right)^{1+s}}.
\end{align}

\begin{figure}
    \centering
    \includegraphics[width=0.7\textwidth]{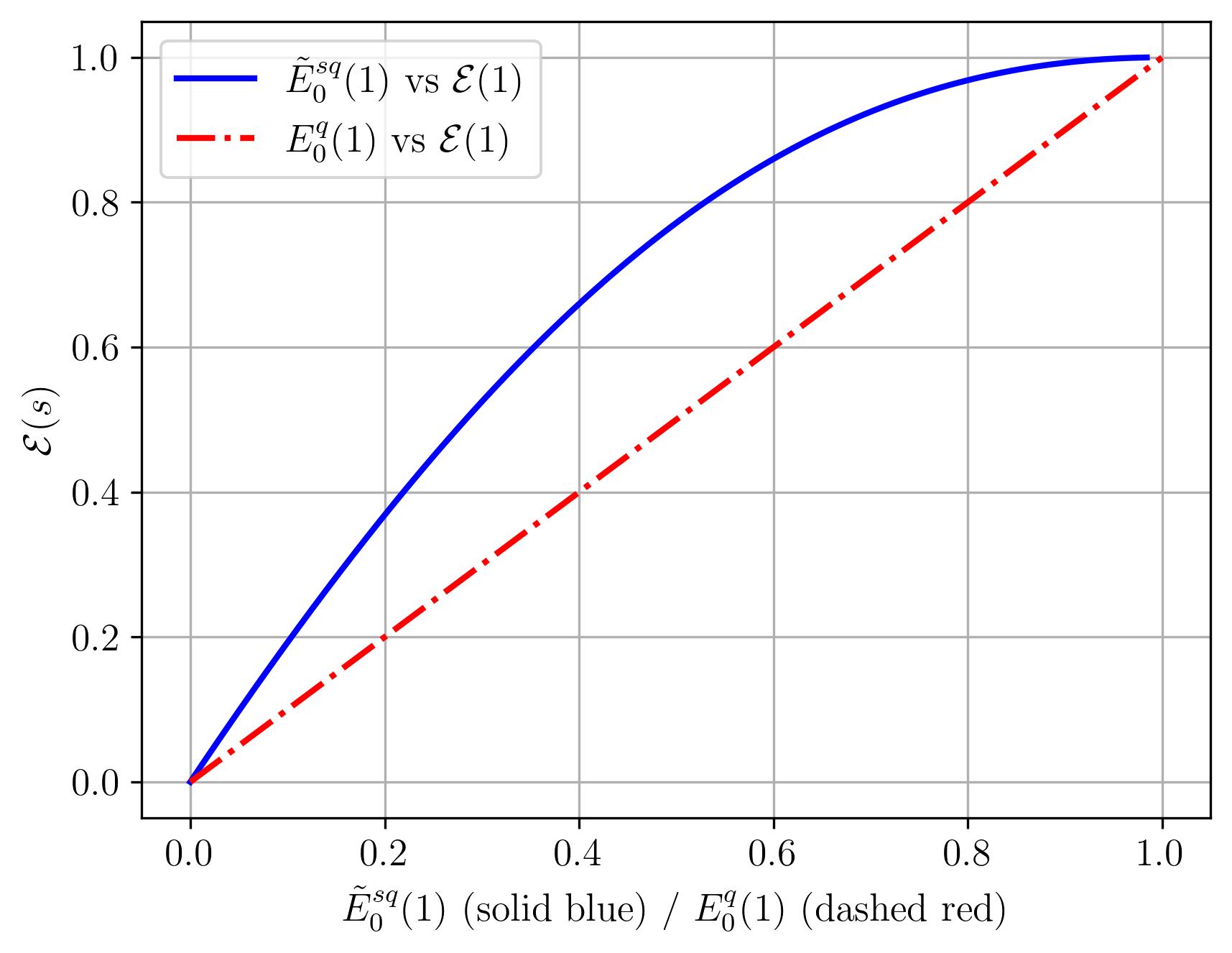}
    \caption{Comparison of the bounds $\mathcal{E}(1)$, Eq. \eqref{eq:SemiQuantumGallagerFunc}, $\tilde{E}_0^{sq}(1)$, Eq. \eqref{eq:cutoffSQ} and $E_0^q(1)$, Eq. \eqref{eq:gallagerfuncBH}, for different pure-state binary channels, defined in Eq. \eqref{eq:binary}. The bound $\mathcal{E}(1)$ is always greater than or equal to $\tilde{E}_0^{sq}(1)$, and coincides with $E_0^q(1)$ for all channels.}
    \label{fig:BoundComp}
\end{figure}

Fig. \ref{fig:BoundComp} displays the dependence between our bound $\mathcal{E}(1)$ and $\tilde{E}_0^{sq}(1)$ for all different classical-quantum binary channels, defined by the correspondence:
\begin{align}\label{eq:binary}
    x_i \to \rho_i=\ket{\psi_i}\!\bra{\psi_i}, \ \ i\in\{0,1\}.
\end{align}
As we can see, it holds $\mathcal{E}(1)\geq\tilde{E}_0^{sq}(1)$. 

On the other hand, in Fig \ref{fig:BoundComp}, we also include the values of the bound $E_0^q(1)$ corresponding to the cutoff in the joint measurement case, Eq. \eqref{eq:gallagerfuncBH}. For all pure-state binary channels, our bound obtained employing the $\alpha$-$z$-R{\'e}nyi relative entropies satisfies $\mathcal{E}(1)=E_0^q(1)$. Moreover, it can be demonstrated that, in this case, $\mathcal{E}(s)=E_0^q(s)$ for all $s\in[0,1]$.

\section{Concluding Remarks}

By revisiting the single-letter communication bounds provided by the generalized Holevo theorem~\cite{bussandri2020generalized}, we have focused on the implications 
of using the $\alpha$-R{\'e}nyi divergences as distance quantifiers in this context. Therefore, we have obtained the \textit{Holevo-R{\'e}nyi inequality}, which stands for an upper bound for the $\alpha$ -R{\'e}nyi divergence between the joint and separable probability distributions, Eq. \eqref{eq:jointprob} and \eqref{eq:separableprobdistr}, respectively (see Theorem \ref{theo:MAIN}), obtaining additionally its tightest bounds, Eqs. \eqref{eq:tightestbounds} and \eqref{eq:tightestbounds2}.

We have delved into practical scenarios where the Holevo-R{\'e}nyi inequality finds relevance: It provides a bound for the $\alpha$-mutual information~\cite{sibson1969information} and, correspondingly, to the channel's capacity of order $\alpha$, a quantity which has gained recent attention in the context of the resource theory of measurement informativeness~\cite{Skrzypczyk2019,Ducuara2022}.

Furthermore, we have also addressed the concept of reliability functions and error exponents in the context of classical-quantum communication channels, deriving a quantum bound $\mathcal{E}(s)$ for the Gallager's reliability function, Eq. \eqref{eq:gallagerfunc}, for the semi-quantum case, in which the classical information is encoded into a set of quantum states ($x \to \rho_x$) and the information is retrieved by the receiver by taking individual measurements over each signal state. We have compared $\mathcal{E}(s)$ to other existing results\cite{burnashev1998reliability,charbit1989cutoff} for the quantum binary channel in the pure state case. 

We have shown therefore that $\mathcal{E}(s)$ limits the quality of the transmission of classical information in the semi-quantum case, in the same way, Holevo's information $C(p)$ bounds the accessible information $I_a(p)$, which for memoryless channels $I_a(p)$ provides the channel capacity: $\tilde{C}_{sq}=\max_p I_a(p)$. Because the Holevo's bound leads to the actual channel capacity when \textit{collective} measurements are allowed $\tilde{C}_q=\max_p C(p)$, one may think that $\mathcal{E}(s)$ would be more representative of the actual reliability function in this case, i.e. for joint POVM performed over $\rho_w=\rho_{x_1}\otimes \cdots \otimes \rho_{x_n}\in B(\mathcal{H}^{\otimes n})$. This statement seems to be supported by Fig. \ref{fig:BoundComp} and by the fact that $\mathcal{E}(s)$ results to be equal to the reliability function demonstrated by Burnashev and Holevo~\cite{burnashev1998reliability}, for binary quantum channels in the case of pure signal. 
%

\section*{Acknowledgments}

\noindent 
D.G.B. and P.W.L. 
are grateful to the Jagiellonian University for the hospitality during their stay in Cracow. They acknowledge financial support from Consejo Nacional de Investigaciones Cient{\'i}ficas y T{\'e}cnicas (CONICET), and by Universidad Nacional de C{\'o}rdoba (UNC), Argentina. K.{\.Z}. is supported by Narodowe Centrum Nauki under the Quantera project number 2021/03/Y/ST2/00193.
G.R.-M. 
acknowledges support
from ERC AdG NOQIA; MICIN/AEI (PGC2018-0910.13039/501100011033, CEX2019-000910-S/10.13039/501100011033, Plan National FIDEUA PID2019-106901GB-I00, FPI; MICIIN with funding from European Union NextGenerationEU (PRTR-C17.I1): QUANTERA MAQS PCI2019-111828-2); MCIN/AEI/ 10.13039/501100011033 and by the “European Union NextGeneration EU/PRTR" QUANTERA DYNAMITE PCI2022-132919 within the QuantERA II Programme that has received funding from the European Union’s Horizon 2020 research and innovation programme under Grant Agreement No 101017733 Proyectos de I+D+I “Retos Colaboración” QUSPIN RTC2019-007196-7); Fundació Cellex; Fundació Mir-Puig; Generalitat de Catalunya (European Social Fund FEDER and CERCA program, AGAUR Grant No. 2021 SGR 01452, QuantumCAT \ U16-011424, co-funded by ERDF Operational Program of Catalonia 2014-2020); Barcelona Supercomputing Center MareNostrum (FI-2023-1-0013); EU (PASQuanS2.1, 101113690); EU Horizon 2020 FET-OPEN OPTOlogic (Grant No 899794); EU Horizon Europe Program (Grant Agreement 101080086 — NeQST), National Science Centre, Poland (Symfonia Grant No. 2016/20/W/ST4/00314); ICFO Internal “QuantumGaudi” project; European Union’s Horizon 2020 research and innovation program under the Marie-Skłodowska-Curie grant agreement No 101029393 (STREDCH) and No 847648 (“La Caixa” Junior Leaders fellowships ID100010434: LCF/BQ/PI19/11690013, LCF/BQ/PI20/11760031, LCF/BQ/PR20/11770012, LCF/BQ/PR21/11840013). 
Views and opinions expressed are of the authors only and do not reflect those of the European Union, European Commission,
nor any other granting authority. 
Neither the European Union nor any granting authority can be held responsible for them.

\renewcommand\bibname{References}
	\bibliographystyle{ws-ijqi}
	\bibliography{library3}

\end{document}